\documentclass[11pt,oneside]{article} 

\usepackage{a4wide}
\usepackage[title]{appendix}

\usepackage{amsmath}
\usepackage{color}
\usepackage{framed}
\usepackage{tikz}
\usepackage{tikz-cd}

\RequirePackage{amsmath}
\RequirePackage{amssymb}
\RequirePackage{amsthm}
\RequirePackage{color}
\RequirePackage{url}
\RequirePackage{mdwlist}
\RequirePackage{arydshln}

\RequirePackage{rotating}

\RequirePackage[all]{xy}
\RequirePackage{graphicx}
\RequirePackage{subcaption}

\usepackage{xcolor}
\usepackage{amsmath,amsfonts,amssymb}
\usepackage{graphicx}
\usepackage{enumerate}
\usepackage{mathrsfs}
\usepackage{mathtools}
\usepackage{authblk}

\makeatletter
\newcommand{\verbatimfont}[1]{\renewcommand{\verbatim@font}{\ttfamily#1}}
\makeatother

\newcommand{\GL}{\mathrm{GL}}
\newcommand{\GU}{\mathrm{GU}}

\newcommand{\Sp}{\mathrm{Sp}}

\newcommand{\U}{\mathrm{U}}

\newcommand{\Field}{\mathbb{F}}
\newcommand{\GF}{\mathrm{GF}}

\newcommand{\tensor}{\otimes}

\usepackage{mathtools}

\newcounter{numitem}

\newcommand{\ket}[1]{|{#1}\rangle}

\newtheorem{lemma}{Lemma}[section]
\newtheorem{proposition}[lemma]{Proposition}
\newtheorem{corollary}[lemma]{Corollary}
\newtheorem{example}[lemma]{Example}
\newtheorem{theorem}[lemma]{Theorem}
\newtheorem{definition}[lemma]{Definition}

\usepackage{listings}

\newcommand{\OO}{\mathrm{O}}
\newcommand{\FF}{\mathbb{F}_2}

\newtheorem{ex}{Example}[section]

\newcommand{\Moiio}{\begin{pmatrix}0&1\\1&0\end{pmatrix}}
\newcommand{\Moiii}{\begin{pmatrix}0&1\\1&1\end{pmatrix}}
\newcommand{\Miooo}{\begin{pmatrix}1&0\\0&0\end{pmatrix}}
\newcommand{\Miooi}{\begin{pmatrix}1&0\\0&1\end{pmatrix}}
\newcommand{\Mioio}{\begin{pmatrix}1&0\\1&0\end{pmatrix}}

\newcommand{\Miioo}{\begin{pmatrix}1&1\\0&0\end{pmatrix}}

\title{A Classification of Transversal Clifford Gates for Qubit Stabilizer Codes}

\author{Shival Dasu${}^1$ and Simon Burton${}^2$ \\
{\small
${}^1$Quantinuum, 303 South Technology Ct., Broomfield, CO 80021, USA \\
${}^2$Quantinuum, Terrington House, 13–15 Hills Road, Cambridge, CB2 1NL, UK
}
}
\date{\today}
\begin{document}

\maketitle
\begin{abstract}
    This work classifies stabilizer codes by the set of diagonal Clifford gates that can be implemented transversally on them. We show that, for any stabilizer code, its group of diagonal transversal Clifford gates on $\ell$ code blocks must be one of six distinct families of matrix groups. We further develop the theory of classifying stabilizer codes by via matrix algebras of endomorphisms first introduced by Rains \cite{Rains1999}, and give a complete classification of the diagonal Clifford symmetries of $\ell$ code blocks. A number of corollaries are given in the final section.
\end{abstract}

\section{Introduction}

Quantum error correction is widely regarded as
essential for the reliable operation of quantum computers.
A variety of quantum error correcting
codes have been developed to protect
quantum information from noise, with new codes continuing
to emerge. Yet, the core challenge in fault-tolerant
quantum computing often lies not in constructing robust
quantum memories, but in enabling logical operations
on encoded qubits in a fault-tolerant manner—ensuring
that errors do not proliferate and degrade the logical
fidelity. As such, methods for characterizing the logical
gate sets supported by a code are critical to advancing
fault-tolerant architectures. 

\emph{Transversal gates} are inherently fault-tolerant because 
they act locally on each code block, while preserving 
the codespace.
We recall some well-known examples.
For CSS codes \cite{Calderbank1996,Steane1996}
we have a transversal CNOT across
$\ell=2$ copies of the code.
Moreover, this operation is transversal \emph{if and only if}
the code is CSS~\cite{Gottesman_1998}.
And so we have characterized the CSS structure of the code 
via the transversal gates supported therein.
Another example comes from $GF(4)$-linear codes \cite{Calderbank1998}:
these support a transversal \emph{facet} gate $SH$ across
a single ($\ell=1$) codeblock.
Once again this implication goes both ways.
Finally, a code is self-dual and CSS if and only if
the entire single qubit Clifford group 
acts as a logical
transversal gate on $\ell=1$ codeblock \cite{Gottesman_1998}.


In each of these cases we find an \emph{operational}
characterization of a class of codes.
The next question to ask is, have we found all instances of
this phenomenon? 
More specifically, what is the group of
\emph{diagonal Clifford gates} on $\ell$ code blocks of a given
stabilizer code (see Fig~\ref{fig:diagonal})?
And then what is the corresponding class of codes?

\begin{figure}
\begin{center}
\includegraphics[]{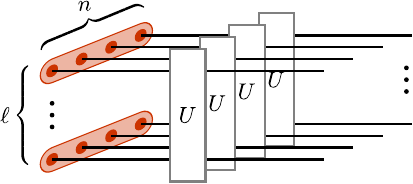}
\end{center}
\caption{
We consider $\ell$ copies of an $n$-qubit code.
For clarity here we choose $n=4$.
A diagonal Clifford operator uses $n$ copies of
the same Clifford gate $U$.
}\label{fig:diagonal}
\end{figure}

In this work we completely solve this problem, 
finding that the group of transversal diagonal $\ell$-qubit Clifford
gates on $\ell$ code blocks of a stabilizer code must
be one of six families of matrix groups, 
together with their corresponding code families.
See Fig~\ref{fig:family}. 


The foundations for these results come from 
early work by Rains \cite{Rains1999},
relating 
the endomorphism algebra $A$ of a code $C$
to operations on $\ell$ code blocks $C^{(\ell)}$.
Rains proves that if a unitary
operator isn't realized by a ring of matrices over $A$ (to be made
precise later), then there exists an ``$A$-linear code''
for which the operator isn't a transversal gate. We strengthen
this by proving that this operator isn't a valid transversal
gate for all $A$-linear codes $C$ 
such that $A$ coincides with 
the algebra of $\Field_2$-linear endomorphisms of $C.$ 
This makes our classification exact.
We also classify up to local Clifford
equivalence all the possible algebras in the qubit case
to find the six families above.
See Fig~\ref{fig:algebras}.

We give a brief outline of the paper. In Section \ref{sec:preliminaries}, we introduce the stabilizer formalism, presented in terms of linear algebra over $\mathbb{F}_2$. In Section \ref{sec:Transversal Algebras}, we define an algebra of transversal endomorphisms of a stabilizer code and give examples. In Section \ref{sec:Classification of Algebras}, we classify the possible endomorphism algebras of any stabilizer code, and, in Section \ref{sec:Transversal Algebra on l blocks}, we show how this determines the endomorphism algebra on $l$ code blocks. In Section \ref{section:Gate Classification}, we prove the main theorem, Theorem \ref{thm:main theorem}, which classifies all stabilizer codes by their group of transversal Clifford gates.

In Section \ref{sec:Ex}, we give a number of applications
of Theorem \ref{thm:main theorem}.
As an example of how our result can advance
fault-tolerant architectures, this includes an application
to magic state distillation and injection. We discuss
how this was applied to a magic state protocol recently
introduced by the authors in \cite{dasu2025}, which was
used to experimentally demonstrate a break-even non-Clifford
gate. This shows how having a classification of the Clifford
gates of a code can even lead to low-overhead protocols
even for non-Clifford gates, thus greatly aiding in the
design of schemes for universal fault-tolerant computation.

\begin{figure}
$$
\begin{tikzcd}[every arrow/.append style={dash},column sep=tiny]
  &       \text{self-dual CSS} &     \\
\GF(4)\arrow[ur]   & \text{semi-self-dual CSS}\arrow[u]\arrow[r,"\cong"] &  \text{self-dual semi-CSS} \arrow[ul]      \\
  &       \text{CSS}\arrow[u] &  \text{self-dual}\arrow[u] \\
\text{(i)}  &       \text{generic}\arrow[uul]\arrow[u]\arrow[ur] &            
\end{tikzcd}
\begin{tikzcd}[every arrow/.append style={dash},column sep=tiny]
  &       \Sp(2\ell,\Field_2) &     \\
\GU(\ell,\Field_2) \arrow[ur]   & \U(\ell,R_8)\arrow[u]\arrow[r,"\cong"] &      \U(\ell,R_8) \arrow[ul]      \\
  &       \GL(\ell,\Field_2)\arrow[u] & \OO(\ell, \Field_2[x]/(x^2))\arrow[u] \\
\text{(ii)}  &       \OO(\ell,\Field_2)\arrow[uul]\arrow[u]\arrow[ur] &       
\end{tikzcd}
$$
\caption{ 
(i) Each family of stabilizer code is characterized by
(ii) a corresponding group of 
diagonal Clifford operators on $\ell$ codeblocks.
The connecting lines show the lattice structure, of
(i) code family inclusions or (ii) subgroups.
}\label{fig:family}
\end{figure}
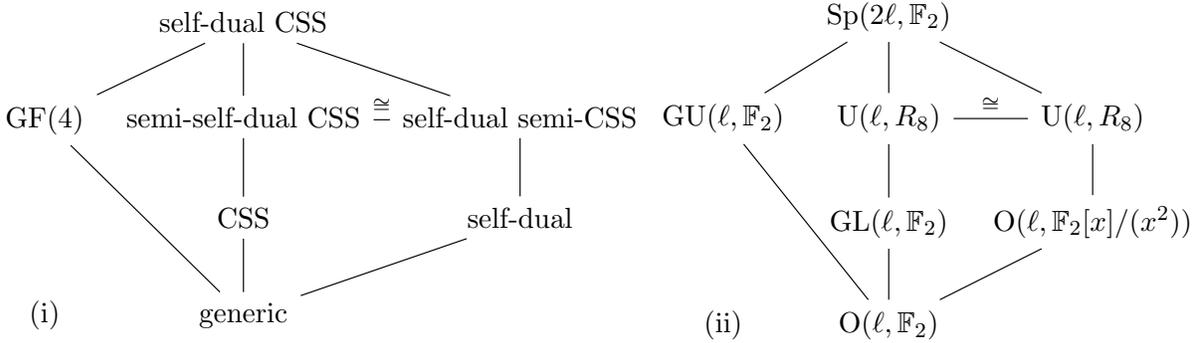

\section{Preliminaries}\label{sec:preliminaries}



\subsection{The Pauli and Clifford Groups}
We will first describe the stabilizer formalism \cite{Gottesman1997}. The Pauli matrices are $2 \times 2$ matrices over the complex numbers defined by 
\begin{equation}
    I = \begin{pmatrix}
        1 & 0 \\
        0 & 1
    \end{pmatrix}, \quad     X = \begin{pmatrix}
        0 & 1 \\
        1 & 0
    \end{pmatrix}, \quad Y = \begin{pmatrix}
        0 & -i \\
        i & 0
    \end{pmatrix}, \quad Z = \begin{pmatrix}
        1 & 0 \\
        0 & -1
    \end{pmatrix}.
\end{equation}
Note that the $X$ and $Z$ matrices anti-commute: $XZ$ = $-ZX$ and $X^2 = Y^2 = Z^2 = I$. The Pauli group on $\ell$ qubits, $\mathcal{P}_\ell$, is generated by all $\ell$-fold tensor products of the above operators.

The Clifford group on $\ell$ qubits is the set of all unitaries which preserve $\mathcal{P}_\ell$ under conjugation. 

\begin{definition}
    The Clifford group on $\ell$ qubits is defined by
    \begin{equation*}
        Cl_\ell = \bigl\{U \in U(2^\ell) | U \mathcal{P}_\ell U^\dagger
            \subseteq \mathcal{P}_\ell\bigr\},
    \end{equation*} where $U(2^\ell)$ is the complex unitary group on $\mathbb{C}^{2^\ell}$. This group is infinite because it is closed under multiplication by unit-modulus scalars. To get a finite group, it is common to define the Clifford group mod scalars 
    \begin{equation}
        \text{Cliff}_\ell := \bigl\{U \in U(2^n) | U \mathcal{P}_\ell
        U^\dagger \subseteq \mathcal{P}_\ell\bigr\}/U(1).
    \end{equation}
\end{definition}

\subsection{The Stabilizer Formalism and Symplectic Linear Algebra}

We can also remove scalars from the Pauli group. Because
$XZ = -iY$, $\mathcal{P}_n$ contains multiples of $i$.
In fact, the diagonal matrices in $\mathcal{P}_n$ are
$i^lI_{2^n}$, $0 \leq l \leq 3$, which is a cyclic group
of order 4, $C_4 = \langle iI_{2^n}\rangle$. If we mod out by these
phases, the Pauli group becomes a 2n-dimensional vector
space over $\mathbb{F}_2 \cong \mathbb{Z}/2\mathbb{Z}$
the finite field with 2 elements \cite{mastel2023cliffordtheorynqubitclifford}.

\begin{proposition}
    Let $\widetilde{\mathcal{P}_n} := \mathcal{P}_n\big/\langle iI_{2^n}\rangle.$
We have that $$\widetilde{\mathcal{P}_n} \cong \mathbb{F}_2^{2n}$$
    where the isomorphism sends a Pauli string
$X_1^{v_1}Z_1^{w_1} X_2^{v_2} Z_2^{w_2}...X_n^{v_n}Z_n^{w_n}$ to the row vector 
$\bigl((v_1, w_1), (v_2, w_2), ... ,(v_n,w_n)\bigr)$.
\end{proposition}
Two Pauli strings in $\mathcal{P}_n$ either commute or anti-commute depending on the number of qubits on which one has an $X$ and the other has a $Z$. For instance $X_1Z_2$ commutes with $Z_1X_2$ because there are two anti-commmuting $X$/$Z$ pairs which give canceling negative signs. We give a simple way of expressing the commutation relation in terms of linear algebra over $\mathbb{F}_2.$ 

Let $M_n(\mathbb{F}_2)$ be the ring of $n \times n$ matrices with entries in $\mathbb{F}_2$, i.e., matrices of $0$s and $1$s. Let $J \in M_2(\mathbb{F}_2)$ be the matrix $J := \begin{pmatrix} 0 & 1 \\ 1 & 0 \end{pmatrix}$. We will also define the matrices $J_n \in M_{2n}(\mathbb{F}_2)$ as block diagonal matrices with $J$s on the diagonal:
\begin{equation}
    J_n := \begin{pmatrix}
    J & \dots & 0 \\
    \vdots & \ddots & \vdots \\
    0 & \dots & J \\
\end{pmatrix}.
\end{equation}
This matrix allows us to define a symplectic inner product on $\mathbb{F}_{2}^{2n}$ which is 1 exactly when the corresponding Pauli strings anti-commute.
\begin{definition}
    Given two row vectors $a,c \in F_2^{2n}$, we define the symplectic inner product $\omega: \mathbb{F}_2^{2n} \times \mathbb{F}_2^{2n} \rightarrow \mathbb{F}_2$ by
    \begin{equation}
        \omega(a,c) := a J_nc^T.
    \end{equation}
\end{definition}

That $\omega$ is an inner product simply means that it is bilinear and non-degenerate; that it is symplectic means that $\omega(v,v) = 0$ for all vectors $v$. 

\begin{proposition}\label{prop: symplectic commutation}
    For any two Pauli strings $X_1^{a_1}Z_1^{b_1}...X_n^{a_n}Z_n^{b_n}$ and $X_1^{c_1}Z_1^{d_1}...X_n^{c_n}Z_n^{d_n}$ in $\mathcal{P}_n$ with exponents given by row vectors 
$a = \bigl((a_1,b_1), ...,(a_n,b_n)\bigr)$ and 
$c = \bigl((c_1,d_1), ...,(c_n,d_n)\bigr)$ respectively, we have
    
     $$(X_1^{a_1}Z_1^{b_1}...X_n^{a_n}Z_n^{b_n}) (X_1^{c_1}Z_1^{d_1}...X_n^{c_n}Z_n^{d_n}) = (-1)^{\omega(a,c)}(X_1^{c_1}Z_1^{d_1}...X_n^{c_n}Z_n^{d_n})(X_1^{a_1}Z_1^{b_1}...X_n^{a_n}Z_n^{b_n}).$$
\end{proposition}

A stabilizer code can be thought of as a subgroup of $G \subset \mathcal{P}_n$ of size $2^{n-k}$. It is often presented by giving $n-k$ ``independent," commuting Pauli strings,  
$G = \langle s_1,...,s_{n-k}\rangle$.
We can make sense of the independence condition in light of our translation to $\mathbb{F}_2$: the images of the $s_i$ in $\mathbb{F}_2^{2n}$ should be linearly independent vectors, and the image of G is a vector subspace of $\mathbb{F}_2^{2n}$ of dimension $n-k$. The fact that the stabilizers commute is equivalent to any two vectors, $a$ and $c$ in $G$, being orthogonal with respect to $\omega$, i.e., $\omega(a,c) = 0$ by Proposition \ref{prop: symplectic commutation}. Since $\omega$ is a symplectic inner product, we call vector subspaces that are self-orthogonal under $\omega$ symplectic.  If $G$ has dimension $n-k$, the Hilbert space of states fixed by $G$ will have dimension $2^{k}$, i.e., will encode $k$ qubits \cite{Aaronson2004}.  We call this subspace $\mathcal{H}_G$.

\begin{definition}\label{def:stab code}
    An $[[n,k,d]]$ stabilizer code $C$ is an $n-k$ dimensional vector subspace of $\mathbb{F}_2^{2n}$ that is symplectic with respect to $\omega.$ Furthermore, the minimum weight of any codeword in $C^\perp\backslash C$ is $d.$
\end{definition}

For an elaboration of the correspondence between $[[n,k,d]]$ stabilizer codes and $\FF$ vector spaces, see \cite{Calderbank1998}.

We now translate Clifford operators into this language. Given a Clifford operator $U$, by definition, conjugation will send any Pauli string $P \in \mathcal{P}_\ell$ to some other Pauli string $UPU^\dagger$. Viewing these Pauli strings as vectors in $\mathbb{F}_2^{2\ell},$ this induces a map from $\mathbb{F}_2^{2\ell}$ to itself. Since $UP_1 P_2U^\dagger = UP_1U^\dagger UP_2U^\dagger$, this map is linear. Also, since conjugation preserves commutation relations, this linear map is symplectic. In fact, these symplectic maps are in one-to-one correspondence with Clifford operators if we mod out by $\widetilde{\mathcal{P}}_\ell.$

\begin{proposition}\label{prop: Clifford Tableau isomorphism}
    Sending a Clifford operator to the map it induces on $\mathbb{F}_2^{2\ell}$ through conjugation of $\mathcal{P}_\ell$ yields the following isomorphism
    \begin{equation}
        \text{Cliff}_{\ell}/\widetilde{\mathcal{P}_\ell} \cong Sp(2\ell, \mathbb{F}_2).
    \end{equation}
\end{proposition}

We will not prove this proposition (for the proof see \cite{mastel2023cliffordtheorynqubitclifford}) but we will give explicitly the map from $\text{Cliff}_{\ell}$ to $Sp(2\ell, \mathbb{F}_2)$. This map sometimes called the Clifford tableau \cite{Gidney2022}.

\begin{definition}
    For an $\ell$ qubit Clifford $U \in \text{Cliff}_\ell$, we will represent the Clifford tableau of U as the $2\ell \times 2 \ell$ symplectic matrix $T$ with entries $t_{r,s} \in \mathbb{F}_2$, $1 \leq r,s\leq 2\ell$ defined by $UX_iU^\dagger = (-1)^c\prod_{j=1}^{\ell} X_j^{t_{2i-1,2j-1}}Z_j^{t_{2i-1,2j}}$ and $UZ_iU^\dagger = (-1)^d\prod_{j=1}^{\ell} X_j^{t_{2i,2j-1}}Z_j^{t_{2i,2j}}$ for some constants $c,d \in \mathbb{F}_2$.
\end{definition}

The definition above is just saying that the rows of $T$ give us the images of $X_i$ and $Z_i$ under conjugation by $U.$ We ignore the constants $c$ and $d$ because we have only defined stabilizer codes up to sign in Definition $\ref{def:stab code}.$ This is not an issue because if a transversal gate $U \otimes U ... \otimes U$ preserves a code up to sign, the appropriate destabilizers can be added to $U$ to ensure that $P_1 U \otimes P_2U ... \otimes P_\ell U$ preserves the sign of the stabilizers as well, where the $P_i$ are Pauli operators. For more information about destabilizers, see \cite{Aaronson2004}.

Therefore, we define transversal gates in terms of Clifford tableaus. For a code $C \subset \FF^{2n}$, the stabilizers of $\ell$ copies of $C$ will be the vector space $C^{(\ell)} = C\oplus ...\oplus C$, the $\ell$-fold direct sum of $C$. Similarly, if $T$ is the Clifford tableau of $U$, then $T\oplus ... \oplus T$, the $n$-fold direct sum of $T$, will be the Clifford tableau of $U^{\otimes n}$.

\begin{definition}
    Let $C \subset \FF^{2n}$ be a stabilizer code. We define the group of $\ell$-qubit transversal gates on $\ell$ copies of C, $G_C^{\ell},$ to be the $2\ell \times 2\ell$ Clifford Tableaus $T \in Sp(2\ell, \mathbb{F}_2)$ such that $T^{\oplus n}$ sends $C^{(\ell)}$ to itself.
\end{definition}

\section{Stabilizer Codes Invariant Under Algebras}\label{sec:Transversal Algebras}

We will now go one step further and see that not only are codes symplectic vector spaces over $\FF,$ but they are actually symplectic modules over $\FF$-algebras. This idea was first introduced in \cite{Rains1999}.

We can define a ``transversal action'' of any $2 \times 2$ matrix, $T$, 
over $\mathbb{F}_2$ on $v$ in $\mathbb{F}_2^{2n}$ by
letting $T$ act by matrix multiplication on the right
on the consecutive $2 \times 1$ row vectors of $v.$

\begin{definition}\label{def:transversal action}
    For any $2 \times 2$ matrix $T \in M_2(\mathbb{F}_2)$ and $v \in \mathbb{F}_2^{2n}$, we let
    \begin{equation}
        \bigl((v_1,w_1), (v_2,w_2) ... (v_n,w_n)\bigr) \cdot T := 
        \bigl((v_1,w_1) T, (v_2,w_2) T, ... (v_n,w_n)T\bigr)
    \end{equation}
\end{definition}

\begin{definition}\label{def:invariant under matrix}
We say that a code $C$ is invariant under a matrix $T\in M_2(\Field_2)$
if, for any $v \in C$, $v \cdot T \in C.$
\end{definition}

Some of the $2 \times 2$ matrices will correspond to transversal single-qubit unitaries acting on the code, in fact, by Proposition \ref{prop: Clifford Tableau isomorphism} these will be exactly those matrices in $Sp(2, \FF)$ that preserve the code. However, we will see that there are matrices which do not correspond to unitaries which may preserve the code in the sense of Definition \ref{def:invariant under matrix}.

\begin{example}\label{ex:self-dual}
    Recall that the matrix $J$ was defined as $\begin{pmatrix}
        0 & 1 \\
        1 & 0
    \end{pmatrix}$. We let 
    \begin{equation*}
         C = \Bigl\langle 
        \bigl((1,0), (1,0), (1,0), (1,0)\bigr), 
        \bigl((0,1), (0,1), (0,1), (0,1)\bigr) \Bigr\rangle
    \end{equation*}
   In other words, $C$ is the $2$-dimensional subspace of
$\mathbb{F}_2^8$ generated by $\bigl((1,0), (1,0), (1,0), (1,0)\bigr)$
and $\bigl((0,1), (0,1), (0,1), (0,1)\bigr).$ 
In terms of Pauli operators, $C$ is generated by $X_1X_2X_3X_4$
and $Z_1Z_2Z_3Z_4.$
This is the familiar $[[4,2,2]]$ code \cite{Gottesman_1998}. 
   
We will check that $C$ is invariant under $J.$ 
It suffices to check that a basis of $C$ is sent to $C$
   under $J$ since the action of $J$ is linear. We check that
   \begin{align*}
       \bigl((1,0), (1,0), (1,0), (1,0)\bigr) \cdot J &= \bigl((1,0)J, (1,0)J, (1,0)J, (1,0)J) \\
       &= \bigl((0,1), (0,1), (0,1), (0,1)\bigr)
   \end{align*} and similarly $\bigl((0,1), (0,1), (0,1), (0,1)\bigr) \cdot J = \bigl((1,0), (1,0), (1,0), (1,0)\bigr).$
   We can interpret this invariance by viewing $J$ as the Clifford tableau of the Hadamard operator. Then, the above equation is simply saying that $H\tensor H \tensor H \tensor H (X_1 X_2 X_3 X_4) (H\tensor H \tensor H \tensor H)^T = 
  Z_1 Z_2 Z_3 Z_4$ and vice-versa. This tells us that the transversal Hadamard operator preserves the code space.
\end{example}

\begin{lemma}
    Let $C$ be any stabilizer code. Let $A \subseteq M_2(\mathbb{F}_2)$ be the set of matrices under which $C$ is invariant. Then A is an algebra over $\FF$.
\end{lemma}

\begin{proof}
$A$ is an algebra over $\FF$ if it is a vector space over $\mathbb{F}_2$ that is closed under multiplication and contains the identity element.
$C$ trivially invariant under $I\in M_2(\FF),$ so $I \in A.$ Since $C$ contains the 0 vector, $0\in A$ as well.
Suppose $m_1, m_2 \in A$. It remains to show that $m_1+m_2 \in A$ and $m_1 m_2 \in A$. For any $v \in C$, we have that $v\cdot(m_1 + m_2) = v\cdot m_1 + v\cdot m_2$, which is in $C,$ since $C$ is closed under addition. Furthermore, we have that $v \cdot (m_1m_2) = (v\cdot m_1)\cdot m_2$, which is also in C.
\end{proof}

\begin{definition}
    For a code $C,$ we call the algebra $A \subseteq M_2(\mathbb{F}_2)$ of $2 \times 2$ matrices under which $C$ is invariant the endomorphism algebra of $C.$
\end{definition}

This algebra $A$ is more accurately called the 
diagonal endomorphism algebra of $C$,
but we abbreviate this to
the endomorphism algebra of $C$.


\addtocounter{ex}{2}
\begin{ex}[continued]
We now compute the endomorphism algebra of the $[[4,2,2]]$
code $C$ above. Because $C$ is self-dual CSS, A contains the
entire Clifford group on 1 qubit, i.e., the symplectic
group of $2\times 2$ matrices over $\FF$: 
$$\Sp(2,\mathbb{F}_2) := \Bigl\{
\begin{pmatrix} 1 & 0 \\ 0 & 1 \end{pmatrix}, 
\begin{pmatrix} 0 & 1 \\ 1 & 0 \end{pmatrix},  
\begin{pmatrix} 1 & 1 \\ 1 & 0 \end{pmatrix}, 
\begin{pmatrix} 0 & 1 \\ 1 & 1 \end{pmatrix},
\begin{pmatrix} 1 & 1 \\ 0 & 1 \end{pmatrix},
\begin{pmatrix} 1 & 0 \\ 1 & 1 \end{pmatrix}
\Bigr\}$$ 
These 6 elements generate the entire set of $2 \times 2$ matrices under addition so $A = M_2(\mathbb{F}_2)$. This reasoning shows that the endomorphism algebra of any self-dual CSS code is $M_2(\mathbb{F}_2)$.
\end{ex}

In the examples below, we use the notation $\mathbb{F}_{2}[x]$
to mean the algebra generated by $x$ and $1,$ the identity.
We will use this notation for $2$-dimensional algebras,
in which case the algebra 
$\mathbb{F}_2[x] = \{0, 1, x, x+1\}$. If $x$ is a $2 \times 2$ matrix, 
this is equal to 
$\Bigl\{
\begin{pmatrix}
    0 & 0 \\
    0 & 0 
\end{pmatrix}, \begin{pmatrix}
    1 & 0 \\
    0 & 1 
\end{pmatrix}, x,
x + \begin{pmatrix}
    1 & 0 \\
    0 & 1 
\end{pmatrix}\Bigr\} .$

\begin{example}\label{F4 linear}
Let $C$ be the standard stabilizers of the $[[5,1,3]]$ code 
\begin{align*}
    C : &= \Bigl\langle
    \bigl((1,0), (0,1), (0,1), (1,0), (0,0)\bigr), 
    \bigl((0,0), (1,0), (0,1), (0,1), (1,0) ) \\
   &\bigl((1, 0), (0,0), (1,0), (0,1), (0,1)\bigr), 
    \bigl((0,1), (1,0), (0,0), (1,0), (0,1)\bigr) \Bigr\rangle
\end{align*}

Let $A_1 = \mathbb{F}_2\Bigl[\begin{pmatrix} 1 & 1 \\ 1 & 0 \end{pmatrix}\Bigr]$. 
Then $A_1$ is a two-dimensional algebra and 
$A_1 = \Bigl\{0, I,\begin{pmatrix} 1 & 1 \\ 1 & 0 \end{pmatrix}, \begin{pmatrix} 0 & 1 \\ 1 & 1 \end{pmatrix}\Bigr\}.$ 
If we let $x = \begin{pmatrix} 1 & 1 \\ 1 & 0 \end{pmatrix},$ 
then $x^2 = x + 1$, so we can define an algebra isomorphism
between $A_1$ and $F_4 = \{0, 1, \alpha, \alpha+1\}$
by sending $I$ to $1$ and $x$ to $\alpha$ ($\mathbb{F}_4$
is the field with $4$ elements defined by adjoining an
element $\alpha$ to $\mathbb{F}_2$ satisfying $\alpha^2 = \alpha + 1$).
Furthermore, $x$ is the well-known facet gate, so if
$C$ is the $[[5,1,3]]$ code, $C \cdot x \subseteq C$,
since a transversal facet gate preserves the stabilizers of this code.
In fact, one can verify that the endomorphism algebra of
the $[[5,1,3]]$ code is exactly $A_1.$ 
\end{example}

Because $A_1$ is isomorphic to $\mathbb{F}_4$, which sometimes referred to as the Galois field extension of size 4, codes that are invariant under $A_1$ are known in the literature as GF(4)-linear codes.

\begin{definition}
    We call a code that is invariant under the algebra $A_1 \cong \mathbb{F}_4$ a GF(4)-linear code.
\end{definition}

We can also reformulate the CSS condition in terms of being invariant under an algebra.

\begin{example}\label{ex: CSS}
Let $A_2 = \mathbb{F}_2\Bigl[\begin{pmatrix} 0 & 0 \\ 0 & 1 \end{pmatrix}\Bigr]$. 
Then $A_2$ is a two-dimensional algebra consisting of the matrices $0, I,\begin{pmatrix} 0 & 0 \\ 0 & 1 \end{pmatrix},$ and $\begin{pmatrix} 1 & 0 \\ 0 & 0 \end{pmatrix}.$ A code is preserved by this algebra if and only if it is CSS. If $C$ is invariant under $\begin{pmatrix} 1 & 0 \\ 0 & 0 \end{pmatrix}$, then $\bigl((x_1, z_1), (x_2,z_2), ... (x_n,z_n)\bigr) \in C,$ implies that $\bigl((x_1, z_1), (x_2,z_2), ... (x_n,z_n)\bigr) \cdot \begin{pmatrix} 1 & 0 \\ 0 & 0 \end{pmatrix} =\bigl((x_1, 0), (x_2,0), ... (x_n,0)\bigr) \in C.$ Similarly, invariance under $\begin{pmatrix} 0 & 0 \\ 0 & 1 \end{pmatrix}$ implies the Z-component of any stabilizer is in $C.$
\end{example}

There is a final two-dimensional algebra that we will discuss in this section.

\begin{example}\label{ex: self-dual non-CSS}
    Let $A_3 = \mathbb{F}_2[J].$ Then $A_3$ consists of the matrices $0, I,\begin{pmatrix} 0 & 1 \\ 1 & 0 \end{pmatrix},$ and $\begin{pmatrix} 1 & 1 \\ 1 & 1 \end{pmatrix}.$ We see that $\begin{pmatrix} 1 & 1 \\ 1 & 1 \end{pmatrix}^2 = 0$, and, in fact, $A_3 \cong \FF[x]/(x^2)$, the polynomial ring over $\FF$ such that $x^2 = 0.$
The code 
$C = \Bigl\langle\bigl((1,0), (1,0), (0,1), (0,1)\bigr), 
\bigl((0,1), (0,1), (1,0), (1,0)\bigr)\Bigr\rangle$
has endomorphism algebra $A_3$.
\end{example}

\begin{definition}
    A self-dual code $C$ is any stabilizer code that is invariant under $J$.
\end{definition}

Note that a code with endomorphism algebra $A_3$ is a self-dual non-CSS code, in contrast to the standard [[4,2,2]] code discussed in Example \ref{ex:self-dual}, which is self-dual CSS. We believe that the self-dual non-CSS code family is somewhat overlooked in the literature, and we will see that it has its own interesting group of transversal gates.

\section{Classification of Algebras}\label{sec:Classification of Algebras}

\begin{figure}
$$
\begin{tikzcd}[every arrow/.append style={dash},column sep=tiny]
                &   A_0=M_2(\Field_2) &                     \\
\!\!A_1=\Field_2\Bigl[\Moiii\Bigr] \arrow[ur]&
    \!\!\!\!A_4=\Field_2\Bigl[\Miooo,\Miioo\Bigr] \arrow[u]\arrow[r,"\cong"]&
    \Field_2\Bigl[\Moiio,\Mioio\Bigr]\arrow[ul]\\
& A_2=\Field_2\Bigl[\Miooo\Bigr] \arrow[u] &
    A_3=\Field_2\Bigl[\Moiio\Bigr] \arrow[u]           \\
\text{(i)} & A_5=\Field_2\Bigl[\Miooi\Bigr] \arrow[uul] \arrow[u]\arrow[ur] & 
\end{tikzcd}
\ \ \ \ 
\begin{tikzcd}[every arrow/.append style={dash},column sep=tiny]
&                  M_2(\Field_2)                 &                     \\
 \Field_4 \arrow[ur]   & R_8 \arrow[u]\arrow[r,"\cong"]      &      R_8 \arrow[ul]      \\
& \Field_2\times\Field_2 \arrow[u]  &  \Field_2[x]/(x^2) \arrow[u]           \\
\text{(ii)} & \Field_2 \arrow[uul] \arrow[u]\arrow[ur] &               
\end{tikzcd}
$$
\caption{The endomorphism algebras used in the classification theorem
\ref{algbera classification},
as (i) concrete matrix algebras,
or (ii) abstract $\Field_2$-algebras.
The connecting lines show inclusion of algebras.
}
\label{fig:algebras}
\end{figure}
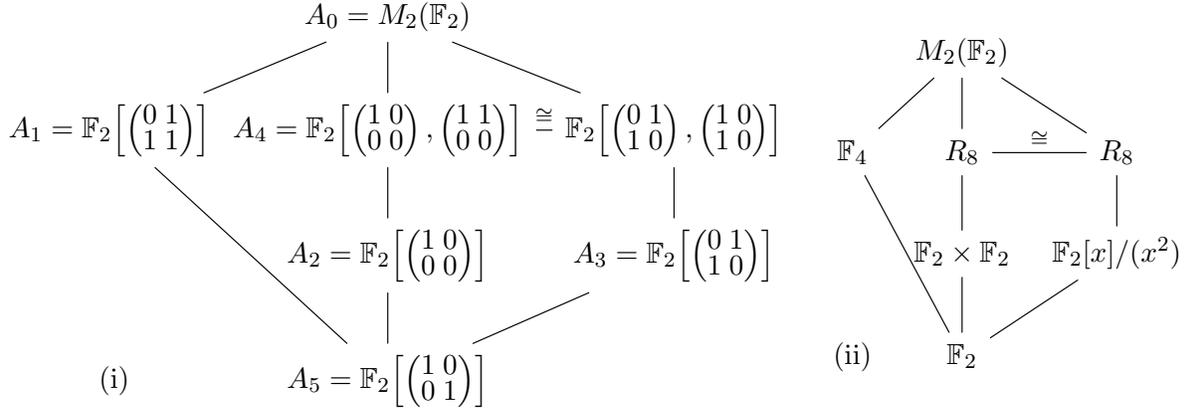

In this section, we classify all the possible endomorphism algebra of a code
can be up to local Clifford equivalence, see Fig~\ref{fig:algebras}.
When we talk
about two codes being locally-Clifford equivalent, we
will restrict our attention only to local-Cliffords of
the form $U\otimes U ...\otimes U$, where each single
qubit Clifford is the same. We call this type of equivalence
a local-diagonal Clifford equivalence (LDC-equivalence).

\begin{theorem}\label{algbera classification}
    
For any code $C \subset \mathbb{F}_2^{2n}$, exactly one
of the following holds:

\noindent (0) The code is LDC-equivalent to a
self-dual CSS code with endomorphism algebra 
$$A_0 = M_2(\mathbb{F}_2).$$ 
\noindent (1) The code is  LDC-equivalent to a
$\GF(4)$-linear code with endomorphism algebra 
$$A_1 = \mathbb{F}_2\Bigl[\begin{pmatrix} 1 & 1 \\ 1 & 0 \end{pmatrix}\Bigr].$$
\noindent (2) The code is LDC-equivalent to a non-self-dual
CSS code with endomorphism algebra 
$$A_2 = \mathbb{F}_2\Bigl[\begin{pmatrix} 1 & 0 \\ 0 & 0 \end{pmatrix}\Bigr].$$
\noindent (3) The code is LDC-equivalent to a self-dual
non-CSS code with endomorphism algebra 
$$A_3 = \mathbb{F}_2\Bigl[\begin{pmatrix} 0 & 1 \\ 1 & 0 \end{pmatrix}\Bigr].$$
\noindent (4) The code is LDC-equivalent to a CSS code
with the X component of the stabilizers contained in
the Z component and with endomorphism algebra $A_4$ the
ring of upper-triangular matrices in $M_2(\mathbb{F}_2)$ 
$$A_4 = \mathbb{F}_2\Bigl[ 
\begin{pmatrix} 1 & 0 \\ 0 & 0 \end{pmatrix},
\begin{pmatrix} 1 & 1 \\ 0 & 0 \end{pmatrix}
\Bigr].$$
\noindent (5) The endomorphism algebra of the code is the trivial one
$$A_5 =  \Bigl\{\begin{pmatrix} 0 & 0 \\ 0 & 0 \end{pmatrix}, 
\begin{pmatrix} 1 & 0 \\ 0 & 1 \end{pmatrix} \Bigr\} \cong \mathbb{F}_2. $$
\end{theorem}


To prove the classification theorem, we need the following proposition.

\begin{proposition}\label{prop:LC isomorphism}
    A code $C_1$ is LDC-equivalent to a code $C_2$ via $U\tensor U ...\tensor U$ if and only if $C_1\cdot R = C_2,$ where R is the Clifford tableau of U. If $C_1$ has endomorphism algebra $B_1,$ then $C_2$ has endomorphism algebra $R^{-1} B_1 R.$
\end{proposition}

\begin{proof}\label{LC equivalence prop}
    The action of $U \otimes U ...\otimes U$ via conjugation sends vectors $v \in \FF^{2n}$ to $v \cdot R$, where $\cdot$ is the transversal action introduced in the preceding section and $R$ is the Clifford tableau of $U$. $C_1$ and $C_2$ are LDC-equivalent, therefore, if and only if $C_1 \cdot R = C_2$. 
    If $C_2 = C_1 \cdot R,$ then the matrices of $M_2(\mathbb{F}_2)$ that preserve $C_2$ are precisely $R^{-1} B_1 R$ where $B_1$ is the endomorphism algebra of $C_1$.
\end{proof}

We can now prove Theorem \ref{algbera classification}.
\begin{proof}

Proposition \ref{prop:LC isomorphism} directly yields that a code, up to LDC-equivalence, falls into at most one of the cases in Theorem \ref{algbera classification} because LDC-equivalences induce isomorphisms of algebras. Therefore, a code cannot be LDC-equivalent to codes in two different cases $(i)$ and $(j)$ since all the $A_i$ and $A_j$ are non-isomorphic. We show now every code falls into one of these cases up to LDC-equivalence.

 Suppose the endomorphism algebra of $C \subset \mathbb{F}_2^{2n}$ is 1-dimensional. This scenario is covered by case (5), since $\mathbb{F}_2$ is the unique 1-dimensional algebra in $M_2(\mathbb{F}_2)$. 

Suppose now that $C$ has an endomorphism algebra which is 2-dimensional.

 One can enumerate all the two dimensional subalgebras
of $M_2(\mathbb{F}_2)$ as follows: 
$A_1 = \FF\Bigl[\begin{pmatrix} 1 & 1 \\ 1 & 0 \end{pmatrix}\Bigr]$,
$A_2 = \mathbb{F}_2\Bigl[\begin{pmatrix} 1 & 0 \\ 0 & 0 \end{pmatrix}\Bigr],$ $A_3 = \FF\Bigl[\begin{pmatrix} 0 & 1 \\ 1 & 0 \end{pmatrix}\Bigr]$,
$B_0 = \mathbb{F}_2\Bigl[\begin{pmatrix} 1 & 1 \\ 0 & 0 \end{pmatrix}\Bigr]$,
$B_1 = \mathbb{F}_2\Bigl[\begin{pmatrix} 0 & 0 \\ 1 & 1 \end{pmatrix}\Bigr]$,
$B_2 = \FF\Bigl[\begin{pmatrix} 1 & 1 \\ 0 & 1 \end{pmatrix}\Bigr]$,
and $B_3 = \FF\Bigl[\begin{pmatrix} 1 & 0 \\ 1 & 1 \end{pmatrix}\Bigr]$.
We can see that this enumeration is complete since each
2-dimensional subalgebra contains two unique elements
apart from 
$\begin{pmatrix} 0 & 0 \\ 0 & 0 \end{pmatrix}$
and 
$\begin{pmatrix} 1 & 0 \\ 0 & 1 \end{pmatrix}$
and the list above contains all $14$ such elements.

 Algebras $A_1,$ $A_2$, and $A_3$ correspond to cases (1), (2), and (3) of the theorem respectively and have been discussed in Examples \ref{F4 linear}, \ref{ex: CSS}, and \ref{ex: self-dual non-CSS}.

The remaining 2-dimensional subalgebras are equivalent to the ones in cases (2) to (4). For instance, we have that $\begin{pmatrix} 1 & 0 \\ 1 & 1 \end{pmatrix}^{-1} B_2 \begin{pmatrix} 1 & 0 \\ 1 & 1 \end{pmatrix} = A_3$, so codes with endomorphism algebra $B_2$ are LDC-equivalent (via a transversal square-root X gate) to a self-dual non-CSS code. Codes with endomorphism algebra $B_3$ are also LDC-equivalent to this case. It can be checked in similar fashion that codes with endomorphism algebras $B_0$ and $B_1$ are LDC-equivalent to non-self-dual CSS codes.

There are three 3-dimensional subalgebras, the ring of upper-triangular matrices, $U= A_4$, lower-triangular matrices, $L$, and matrices with an even number of non-zero entries, $E$. Codes with algebra $A_4$ are LDC-equivalent to codes with algebra $L$ via the Hadamard operation. Similarly, codes with algebra $A_4$ are LDC-equivalent to codes with algebra $E$ via the $\sqrt{X}$ gate. Therefore, all three-dimensional algebras are covered by case (5).

Since $A_4$ contains the element $\begin{pmatrix} 1 & 0 \\ 0 & 0 \end{pmatrix}$, codes with its endomorphisms are CSS. Since it contains $\begin{pmatrix} 0 & 1 \\ 0 & 0 \end{pmatrix}$, the X component of the code must be contained in the Z component--we call codes with this property ``semi-self-dual". Similarly, codes with endomorphism algebra $L$ are CSS codes with Z component contained in the X component. Both of these algebras contain $A_2$. Interestingly, the algebra with even-parity matrices $E$ does not contain $A_2$, so the corresponding codes are not CSS, yet nevertheless have the same group of transversal Cliffords. This is because they have a ``CSS-like" endomorphism, $\begin{pmatrix} 1 & 0 \\ 1 & 0 \end{pmatrix}$, where if we project \textit{both} $X$ and $Z$ onto $X$, then this sends the code to itself. That's why we call such codes semi-CSS in Fig~\ref{fig:family}. Furthermore, $A_3 \subset E$, so such codes are self-dual. All three algebras are isomorphic to $R_8$, the smallest non-commutative ring with 8 elements.

There is only one 4-dimensional subalgebra of $M_2(\FF),$ which is, of course, all of $M_2(\FF).$ In this case, the code is a self-dual CSS code. See Example \ref{ex:self-dual}.
\end{proof}

Since $A_0$ and $A_1$ were unique up to isomorphism as subalgebras of $M_2(\mathbb{F}_2)$, we did not need to find LDC-equivalences in these cases. This yields the following corollary.

\begin{corollary}\label{cor:two unique cases}
    If a code $C$ is LDC-equivalent to a self-dual CSS code, then it is self-dual CSS. If a code is LDC-equivalent to a $GF(4)$-linear code, then it is GF(4)-linear.
\end{corollary}

\begin{figure}
\begin{center}
\begin{tabular}{c|rrrrrrr}
$\ell$ & 1 & 2 & 3 & 4 & 5 & 6 \\
\hline
$S_\ell$ & 1 & 2 & 6 & 24 & 120 & 720 \\
\hline
O($\ell,\Field_2$) & 1 & 2 & 6 & 48 & 720 & 23040 \\
O($\ell, \Field_2[x]/(x^2)$) & 2 & 16 & 384  & 24576 & ... & ...\\
$\U(\ell, R_8)$ & 2 & 48 & 10752 & ...  & ... & ... \\
GL($\ell,\Field_2$) & 1 & 6 & 168 & 20160 & 9999360 & 20158709760 \\
U($\ell,\Field_4$) & 3 & 18 & 648 & 77760 & 41057280 & 82771476480 \\
Sp($2\ell,\Field_2$) & 6 & 720 & 1451520 & 47377612800 & ...  & ...  \\
\end{tabular}
\caption{The entries
in this table show the orders of six families of matrix groups
which are the possible groups of diagonal transversal
Clifford gates on $l$ code blocks of a stabilizer code.
$M(l,R)$ means a group of $l \times l$ matrices with
entries in the ring $R$. Each matrix group satisfies a
common linear algebra condition--O means orthogonal,
U means unitary, GL means general linear, and Sp means symplectic.
$S_\ell$ is the permutation group which acts on the $\ell$ codeblocks,
and is a subgroup of all the others.
}
\label{fig:Family Table}
\end{center}
\end{figure}

\section{The Transversal Algebra on $\ell$ Code Blocks}\label{sec:Transversal Algebra on l blocks}

Let $C^{(\ell)} = C \oplus C ... \oplus C.$ We think
of $C^{(\ell)}$ as being the stabilizers corresponding
to $\ell$ copies of a code $C.$ We are interested in
the set of ``transversal endomorphisms'' of $C^{(\ell)}.$
We will define these to be $2\ell \times 2\ell$ matrices $T$
such that $T$ acting on sets of $\ell$ qubits given by
the $i$th qubit of each code block preserves $C^{(\ell)}$.
In this section, we prove that these multi-block endomorphisms are determined by the endomorphism algebra of $C.$

Formally, an element of $C^{(\ell)}$ is given by 
\begin{equation*}
    \bigl((v_1^{(1)}, w_1^{(1)}), ..., (v_n^{(1)}, w_n^{(1)})\bigr) \oplus ... \oplus \bigl((v_1^{(\ell)}, w_1^{(\ell)}), ..., (v_n^{(\ell)}, w_n^{(\ell)})\bigr)
\end{equation*}
where $\bigl((v_1^{(k)}, w_1^{(k)}), ..., (v_n^{(k)}, w_n^{(k)})\bigr) \in C$ for $1 \leq k \leq \ell.$
We can rewrite the above as a row vector of length $2n \ell$ where we permute the entries so that the $2\ell$ Pauli operators on qubit $i$ of each codeblock are written contiguously: 
\begin{equation*}
    \bigl((v_1^{(1)}, w_1^{(1)}, v_1^{(2)}, w_1^{(2)}, ..., v_1^{(\ell)}, w_1^{(\ell)}), ... , (v_n^{(1)}, w_n^{(1)}, v_n^{(2)}, w_n^{(2)}, ..., v_n^{(\ell)}, w_n^{(\ell)})\bigr).
\end{equation*}

With this indexing, we can now easily describe the transversal action.
\begin{definition}\label{def: transversal action l blocks}
    For a matrix $T \in M_{2\ell}(\FF)$ and a vector $v \in C^{(\ell)},$ we define the transversal action of $T$ on $v$ by
    \begin{align*}
        \bigl((v_1^{(1)}, w_1^{(1)}, &..., 
            v_1^{(\ell)}, w_1^{(\ell)}), ... , 
            (v_n^{(1)}, w_n^{(1)},  ..., v_n^{(\ell)}, 
            w_n^{(\ell)})\bigr) {\cdot} T :=  \\
            &\bigl((v_1^{(1)}, w_1^{(1)}, ..., 
            v_1^{(\ell)}, w_1^{(\ell)}) T, ... , 
            (v_n^{(1)}, w_n^{(1)},  ..., v_n^{(\ell)}, w_n^{(\ell)}) T \bigr).
    \end{align*}
\end{definition}

We will make use of the following lemma. 
\begin{lemma}\label{lem:transversal action}
    For any $2\ell \times 2\ell$ matrix $T$, we can express $T$ as an $\ell \times \ell$ block matrix with $2 \times 2$ entries, $T = (T_{i,j})_{1 \leq i,j \leq \ell}$. Then, for any vector $v^{(1)} \oplus v^{(2)} ...\oplus v^{(\ell)} \in C^{(\ell)}$, we have
    \begin{equation*}
        v^{(1)} \oplus v^{(2)} ...\oplus v^{(\ell)} {\cdot} T =
        \sum_{k=1}^\ell v^{(k)} \cdot T_{k,1}\oplus ... \oplus \sum_{k=1}^\ell v^{(k)} \cdot T_{k,\ell}
    \end{equation*}
    where $v^{(k)} \cdot T_{k,j}$ is the ``transversal action" of $2 \times 2$ matrices in Definition \ref{def:transversal action}.
\end{lemma}
\begin{proof}
    Using our row vector representation of $v$ so that the $i$th qubit Pauli operators of all $\ell$ blocks are contiguous as before, we have 
    \begin{align*}
        &\bigl((v_1^{(1)}, w_1^{(1)}, ..., v_1^{(\ell)}, w_1^{(\ell)}), ... , (v_n^{(1)}, w_n^{(1)},  ..., v_n^{(\ell)}, w_n^{(\ell)})\bigr) {\cdot} T \\
        &= \Bigl(\bigl( \sum_{k=1}^\ell(v_1^{(k)}, w_1^{(k)}) T_{k,1}, ..., 
            \sum_{k=1}^\ell(v_1^{(k)}, w_1^{(k)}) T_{k,\ell}\bigr), ... , 
            \bigl( \sum_{k=1}^\ell(v_n^{(k)}, w_n^{(k)}) T_{k,1}, ..., 
            \sum_{k=1}^\ell(v_n^{(k)}, w_n^{(k)}) T_{k,\ell}\bigr)\Bigr).
    \end{align*}
    Then, reshuffling terms so that the Pauli operators for each codeblock are written together, we have this equals
    \begin{align*}
        &\bigl( \sum_{k=1}^{\ell} (v_1^{(k)}, w_1^{(k)}) T_{k,1} , ... , 
        \sum_{k=1}^{\ell} (v_n^{(k)}, w_n^{(k)}) T_{k,1}\bigr) 
        \oplus ... \oplus
        \bigl( \sum_{k=1}^{\ell} (v_1^{(k)}, w_1^{(k)}) T_{k,\ell} , ... , 
        \sum_{k=1}^{\ell} (v_n^{(k)}, w_n^{(k)}) T_{k,\ell} \bigr) \\
        &= \sum_{k=1}^{\ell} v^{(k)} \cdot T_{k,1}
        \oplus ... \oplus
        \sum_{k=1}^{\ell} v^{(k)} \cdot T_{k,\ell}.
    \end{align*}
\end{proof}

For a code $C$ with endomorphism algebra $A,$ there is a class of $2\ell \times 2\ell$ matrices $T$ in which we are particularly interested, where the $2 \times 2$ blocks of $T$ are in $A.$ 

\begin{definition}
    For $\ell \geq 1$ and an algebra $A$ of $2 \times 2$ matrices
over $\mathbb{F}_2$, we define $M_\ell(A)$ to be the algebra
of the $2\ell \times 2\ell$ matrices, T, over $\Field_2$ such that, when viewed
as a block matrix with $2 \times 2$ entries, $T = (a_{ij})_{1 \leq i,j \leq \ell}$,
all the $a_{ij}$ are in $A.$ 
\end{definition}

\begin{example}
    Take the algebra 
$A_1 = \FF\Bigl[\begin{pmatrix} 1 & 1 \\ 1 & 0 \end{pmatrix}\Bigr]$. 
Then $M_2(A_1)$ contains the matrix $\begin{pmatrix} 1 & 1 & 1 & 0 \\ 1 & 0 & 0 & 1 \\  1 & 0 & 0 & 1 \\ 0&1&1&1\end{pmatrix}$ because the matrices $\begin{pmatrix} 1 & 1 \\ 1 & 0 \end{pmatrix}, \begin{pmatrix} 0 & 1 \\ 1 & 1 \end{pmatrix},$ and $ \begin{pmatrix} 1 & 0 \\ 0 & 1 \end{pmatrix}$ are in $A_1.$ Identifying the elements of $A_1$ with the finite field with $4$ elements as before, we can view $M_2(A_1)$ as being $M_2(\mathbb{F}_4)$, the $2 \times 2$ matrices over the finite field with 4 elements and rewrite this matrix as $\begin{pmatrix} \alpha & 1 \\ 1 & \alpha +1 \end{pmatrix}.$ 
\end{example}

If $A$ is the endomorphism algebra of $C,$ then, as discussed
earlier, $C$ is an $A$-module. We will now see that $C^{(\ell)}$
is an $M_\ell(A)$ module via the transversal action of
$2\ell \times 2\ell$ matrices. In fact, $M_\ell(A)$ is
precisely the endomorphisms of $C^{(\ell)}$ under
this action, which we now prove.

\begin{theorem}\label{thm:ell block transversal algebra}
    Let $C \subset \mathbb{F}_2^{2n}$ be a code with endomorphism algebra $A.$ The algebra of $2\ell \times 2\ell$ matrices $T$ such that $C^{(\ell)} {\cdot} T \subseteq C^{(\ell)}$ is $M_\ell(A).$
\end{theorem}

\begin{proof}
    We start by showing that if $T \in M_\ell(A)$, then $T$ preserves $C^{(\ell)}.$ By Lemma \ref{lem:transversal action}, we have that for $v^{(1)} \oplus ... \oplus v^{(\ell)} \in C^{(\ell)},$ $v^{(1)} \oplus ... \oplus v^{(\ell)} {\cdot} T = \sum_{k=1}^\ell v^{(k)} \cdot T_{k,1},\oplus ... \oplus \sum_{k=1}^\ell v^{(k)} \cdot T_{k,\ell}$. Since $T \in M_{\ell}(A)$, $T_{k,i} \in A,$ so $\sum_{k=1}^\ell v^{(k)} \cdot T_{k,i} \in C$ for all $i$ since $C$ is invariant under $A$ and $C$ is closed under addition. Therefore $\sum_{k=1}^\ell v^{(k)} \cdot T_{k,1},\oplus ... \oplus \sum_{k=1}^\ell v^{(k)} \cdot T_{k,\ell} \in C^{(\ell)}$.

    We must now show that, conversely, if, for all $v^{(1)} \oplus ... \oplus v^{(\ell)} \in C^{(\ell)}$, $v^{(1)} \oplus ... \oplus v^{(\ell)} {\cdot} T \in C^{(\ell)},$ then $T \in M_\ell(A).$ This is the same as showing that all of the $2 \times 2$ of $T$ blocks are in $A$. To show that that $T_{p,i} \in A,$ we observe that if $\sum_{k=1}^\ell v^{(k)} \cdot T_{k,1},\oplus ... \oplus \sum_{k=1}^\ell v^{(k)} \cdot T_{k,\ell} \in C^{(\ell)},$ then in particular, the ith vector in the direct sum, $\sum_{k=1}^\ell v^{(k)} \cdot T_{k,i} \in C$. Suppose that we are testing the vector in $C^{(\ell)}$ given by $v^{(p)} = v,$ an arbitrary vector in $C$ and $v^{(k)} = 0$ otherwise. Then $\sum_{k=1}^\ell v^{(k)} \cdot T_{k,i} = v \cdot T_{p,i}$. Since this must be in $C$ for all $v \in C$ this means that $T_{p,i}$ preserves $C$ and must therefore be in $A,$ completing the proof.

\end{proof}

\section{Transversal Gates Classification}\label{section:Gate Classification}

We now state and prove our main theorem.
Recall that, for a stabilizer code $C \subseteq \mathbb{F}_2^{2n}$, we define the group of $\ell$-qubit
transversal Clifford gates, $G^\ell_C$, to be the subgroup of $Sp(2\ell,\FF)$ with transversal action (Definition \ref{def: transversal action l blocks}) that preserves $C^{(\ell)}$. Recall that LDC-equivalent means local-diagonal Clifford equivalent.

\begin{theorem}\label{thm:main theorem}
For any stabilizer code, $C$, the group of transversal
Cliffords, $G^\ell_C$, on $\ell$ code-blocks of $C$ must be one of the following:

(0) $\Sp(2\ell,\Field_2)$ and the code is self-dual CSS

(1) $\U(\ell,\Field_4)$ and the code is $\GF(4)$-linear

(2) $\GL(\ell,\Field_2)$ and the code is LDC-equivalent 
to a non-self-dual CSS code

(3) $\OO(\ell, \Field_2[x]/(x^2))$ and the code is 
LDC-equivalent to a self-dual non-CSS code.

(4) $U(\ell, R_8)$ and the code is LDC-equivalent
to a CSS code where any X-type stabilizer, if X is exchanged
with Z, is a Z-type stabilizer for the code

(5) $\OO(\ell,\Field_2)$ otherwise.
\end{theorem}

\begin{proof}

Let $C' \subset \FF^{2n}$ be an arbitrary stabilizer code. By Theorem \ref{algbera classification}, the code is LDC-equivalent to a code $C$ with endomorphism algebra $A_i$ for some $0 \leq i \leq 5$, satisfying the property in case (i) (e.g., is self-dual non-CSS). Furthermore, by Corollary \ref{cor:two unique cases}, if the algebra is $A_0$ or $A_1,$ this LDC-equivalence can be taken to be the identity. This LDC-equivalence induces an isomorphism between the groups of transversal Clifford gates on $C'$ and $C,$ so it remains to characterize the set of transversal gates for codes with endomorphism algebra $A_i.$
    
    An element $T$ of $G_C^\ell$ must preserve the stabilizers
of $C^{(\ell)}$ under the transversal action, since this corresponds to conjugation of the stabilizers by a transversal operator. By Theorem \ref{thm:ell block transversal algebra}, this is true if and only
if $T \in M_\ell(A_i)$. Therefore, $G_C^{\ell} = M_\ell(A_i)\cap Sp(2\ell,\FF),$ where we intersect with $Sp(2\ell,\FF)$ because the operators must actually correspond to a Clifford in addition to preserving the stabilizers. $T \in M_\ell(A_i)$ is symplectic if and only if 

    \begin{equation}\label{eqn:symplectic tableau}
       T J_n T^t = J_n.
    \end{equation}

    We can rephrase this condition as 
    \begin{equation}\label{eq:cond}
    J_nT^tJ_n = T^{-1}.
    \end{equation}
    For $a \in A_i$, let us define a conjugation operation on $A_i$ by
    \begin{equation}\label{eq:conj2by2}
        \bar{a} := Ja^tJ.
    \end{equation}
    It is not obvious a priori that $\bar{a}$ is also in $A_i,$ but the reader can verify that this is the case for all of the $A_i.$ Similarly, for any matrix $T \in M_\ell(A_i)$, if 
    $T = (a_{jk})$ with $a_{jk} \in A_i,$ we can define $\bar{T}$ to be the matrix $(\bar{a}_{jk})$. We can define $\bar{T}^{t}$ by
    \begin{equation}\label{eq:conj transpose}
        \bar{T}^{t}_{jk} := \bar{a}_{kj}
    \end{equation} so we are taking the transpose viewing $\bar{T}$ as a matrix over $A_i.$
    
    Then, for any $T \in M_\ell(A_i)$, since $J_n$ is block diagonal with all diagonal entries $J,$ we have that 
    \begin{equation*}
        J_n T^t J_n = \bar{T}^t.
    \end{equation*}
    Therefore $T \in G_C^\ell$ if and only if $T \in M_\ell(A_i)$ and 
    \begin{equation}\label{eq:physicality condition}
    T \bar{T}^t = I.
    \end{equation} The reader can check that the matrices in $M_\ell(A_i)$ satisfying this unitarity condition are precisely the matrix groups in the theorem. We will work out a few of these explicitly in the next section.
\end{proof}

\section{Examples and Applications}\label{sec:Ex}

We use the notation that for an $\ell \times \ell$ matrix $C$ and a $2 \times 2$ matrix $U$ 
\begin{equation*}
C[U] := \begin{pmatrix} c_{1,1}U & \dots & c_{1,\ell}U \\
    \vdots & \ddots & \vdots \\
    c_{\ell,1}U & \dots & c_{\ell,\ell}U \\
    
\end{pmatrix}
\end{equation*}

\begin{example}
    Take $A_1 = \mathbb{F}_2\Bigl[\begin{pmatrix} 0 & 1 \\ 1 & 1 \end{pmatrix}\Bigr]$ as before. Let $x = \begin{pmatrix} 0 & 1 \\ 1 & 1 \end{pmatrix}.$ Then, since $A_1$ is two-dimensional, any element of $A_1$ can be expressed as $a I + b x,$ so any matrix $T \in M_\ell(A_1)$ can be expressed as $A[I] + B[x],$ where $A$ and $B$ are $\ell \times \ell$ matrices. 
    
    We now compute the action of the conjugation operation in \eqref{eq:conj2by2} on $x$. We have $\bar{x} = J \begin{pmatrix} 0 & 1 \\ 1 & 1 \end{pmatrix} J = \begin{pmatrix} 1 & 1 \\ 1 & 0 \end{pmatrix} = x+ I.$ Identifying $A_1$ with $\mathbb{F}_4$ via $I \rightarrow 1$ and $x \rightarrow \alpha$, this corresponds to the relation $\bar{\alpha} = \alpha + 1$ in $\mathbb{F}_4.$ Under this identification, any $A[I] + B[x]$ in $M_\ell(A_1)$ becomes $$\begin{pmatrix} a_{1,1} + b_{1,1}\alpha & \dots & a_{1,\ell} + b_{1,\ell}\alpha \\
    \vdots & \ddots & \vdots \\
    a_{\ell,1} + b_{\ell,1} \alpha & \dots & a_{\ell,\ell} + b_{\ell,\ell}\alpha \end{pmatrix}$$ if we write $x_{j,k} = a_{j,k} + b_{j,k}$  the symplectic condition \eqref{eq:physicality condition} becomes
    \begin{equation*}
        I = \begin{pmatrix} a_{1,1} + b_{1,1}\alpha & \dots & a_{1,\ell} + b_{1,\ell}\alpha \\
    \vdots & \ddots & \vdots \\
    a_{\ell,1} + b_{\ell,1} \alpha & \dots & a_{\ell,\ell} + b_{\ell,\ell}\alpha \end{pmatrix} \begin{pmatrix} \overline{a_{1,1} + b_{1,1}\alpha} & \dots & \overline{a_{\ell,1} + b_{\ell,1}\alpha} \\
    \vdots & \ddots & \vdots \\
    \overline{a_{1,\ell} + b_{1,\ell} \alpha} & \dots & \overline{a_{\ell,\ell} + b_{\ell,\ell}\alpha} \end{pmatrix}
    \end{equation*}
    where the bar denotes conjugation in $\mathbb{F}_4$.
    We have seen that operations in $M_\ell(A_1)$ are symplectic if and only if they are unitary over $\mathbb{F}_4$. Therefore, the matrix group in this case is $U(\ell,\mathbb{F}_4)$. 
\end{example}

Our theorem yields the following.
    \begin{corollary}
        If $C$ has endomorphism algebra $A_1$, i.e., is any non-self-dual GF(4)-linear code, such as the $[[5,1,3]]$ code, then $G_C$ consists of the identity and both the clockwise and counterclockwise transversal facet gate.
    \end{corollary}
    \begin{proof}
        $G_C$ consists exactly of those operations $T \in M_1(A_1) = A_1$ that satisfy the symplectic condition \eqref{eq:physicality condition}. If $T = I, \begin{pmatrix} 0 & 1 \\ 1 & 1 \end{pmatrix},$ or $\begin{pmatrix} 1 & 1 \\ 1 & 0 \end{pmatrix}$, then it can be verified that $T \bar{T} = I,$ where the latter two matrices are the clockwise and counterclockwise facet gates. Since the $0$ matrix doesn't satisfy the symplectic condition, $G_C$ consists exactly of these three matrices.
    \end{proof}

    We can also furnish a simple proof of Gottesman's claim in \cite{GottesmanStabilizerSymmetries} that there are no non-trivial transversal two-qubit Clifford gates for the $[[5,1,3]]$ code.

    \begin{corollary}\label{cor:no tq gate}
        Let $C$ be a code with endomorphism algebra $A_1$ (such as the $[[5,1,3]]$ code). Then there is no non-trivial transversal two qubit Clifford on $C^{(2)}$. More precisely, all two qubit Clifford gates are swaps followed by transversal facet gates applied to each block separately. 
    \end{corollary}

    \begin{proof}
        We know that $G_C^2$ consists of those matrices of $M_2(\mathbb{F}_4)$ such that 
        \begin{equation*}
            \begin{pmatrix} a & b \\ c & d \end{pmatrix} \begin{pmatrix} \bar{a} & \bar{c} \\ \bar{b} & \bar{d} \end{pmatrix} = \begin{pmatrix} 1 & 0 \\ 0 & 1 \end{pmatrix}.
        \end{equation*}
        For any $x \in \mathbb{F}_4$, $x \bar{x} = 1$ if and only if $x \neq 0$, so we have that the right-hand side of the above equation is equal to 
        \begin{equation*}
            \begin{pmatrix} 1_{a \neq 0} + 1_{b \neq 0} & a \bar{c} + b \bar{d} \\ \bar{a} c + \bar{b} d &  1_{c \neq 0} + 1_{d \neq 0} \end{pmatrix}
        \end{equation*}
        so we must have that exactly one of $a$ and $b$ are non-zero and exactly one of $c$ and $d$ are non-zero for this to be the identity matrix. Suppose that $a$ is non-zero, then $b=0$ implies that $a\bar{c} = 0$, which implies that $c = 0$. In that case, the matrix is of the form
        \begin{equation*}
            \begin{pmatrix} a & 0 \\ 0 & d \end{pmatrix}
        \end{equation*}
        meaning we apply either the identity or a transversal facet gate to the first codeblock, depending on $a$ and do the same to the second codeblock, depending on $d.$ If $b \neq 0,$ then $a = 0$ implying that $b\bar{d} = 0$, implying $d=0$, so the matrix is of the form
        \begin{equation*}
            \begin{pmatrix} 0 & b \\ c & 0 \end{pmatrix}
        \end{equation*}
        corresponding to swapping the two code blocks and possibly performing transversal facet gates depending on $b$ and $c.$
    \end{proof}

\begin{example}
    We now examine the non-self-dual CSS case. That the operations in $M_\ell(A_2)$ which satisfy the symplectic condition \eqref{eq:physicality condition} are indexed by $GL_\ell(\mathbb{F}_2)$ was first observed in \cite{Rains1999}. However, we have now shown that an operation is transversal in this case if and only if it is in $GL_\ell(\FF)$.  Let $x =\begin{pmatrix} 0 & 0 \\ 0 & 1 \end{pmatrix}$. So $A_2 = \mathbb{F}_2[x]$. In this case, we actually wish to write an arbitrary element of $M_\ell(A)$ as $A[x] + B[I+x].$ It can be easily verified that $\bar{x} = x+I$ and $\overline{x +I} = x$. Rewriting equation \eqref{eq:cond} in this case  gives us that 
    \begin{equation}\label{eq:css}
        (A[x] + B[I+x])(A^t[I+x] + B^t [x]) = I.
    \end{equation}
    Since $(I+x)x = 0$, \eqref{eq:css} becomes $AB^t[x] + BA^t[I+x] = I.$ Since $x + (I + x) = I$ and $x, I+x$ are linearly independent, we can see that $AB^t[x] + BA^t[I+x] = I$ is satisfied if and only if $A$ is invertible and $B^t = A^{-1}.$ The set of all such invertible $A$ is $GL_\ell(\mathbb{F}_2).$ We now elaborate how this correspondence yields the CNOT gate. 
    
    Let $A = \begin{pmatrix} 1 & 0 \\ 1 & 1 \end{pmatrix}.$ Then $A$ is invertible as it is its own inverse. Therefore, $B^{t} = A,$ so $B = \begin{pmatrix} 1 & 1 \\ 0 & 1 \end{pmatrix}.$ The corresponding element of $M_\ell(A)$ is therefore
    \begin{equation}
        A[x] + B[I+x] = \begin{pmatrix} I & I+x \\ x & I \end{pmatrix} = \begin{pmatrix} 1 & 0 & 1 & 0 \\ 0 & 1 & 0 & 0 \\ 0 & 0 & 1 & 0 \\ 0 & 1 & 0 & 1 \end{pmatrix}
    \end{equation}
    which is the Clifford Tableau of the CNOT. Recall that the conjugate of $X_1$, for instance, is the image of the row vector $(1,0,0,0)$, which is $(1,0,1,0)$ corresponding to $X_1 X_2$.
\end{example}

\begin{example}
    We now examine the self-dual non-CSS case. Suppose $C$ is a code
with endomorphism algebra $A_3.$ We have that 
$A_3 = \Bigl\{0, I, \begin{pmatrix} 0 & 1 \\ 1 & 0 \end{pmatrix}, \begin{pmatrix} 1 & 1 \\ 1 & 1 \end{pmatrix}\Bigr\}$.
If we let $x = \begin{pmatrix} 1 & 1 \\ 1 & 1 \end{pmatrix},$ we see that $x^2 = 0.$ In fact, $A_3$ is isomorphic to the ring $\mathbb{F}_2[x]/(x^2) = \{0, 1, x, x+1\}$ via the mapping which sends $\begin{pmatrix} 1 & 1 \\ 1 & 1 \end{pmatrix}$ to $x$, $\begin{pmatrix} 0 & 1 \\ 1 & 0 \end{pmatrix}$ to $x+1,$ $I$ to 1 and $0$ to 0.

    We can also compute the conjugation operation in \eqref{eq:conj2by2} on $A_3$, and we find that,for instance, $\bar{x} = Jx^tJ = x.$ In fact, the conjugation operation fixes all the elements of $A_3.$ Therefore, the symplectic condition \eqref{eq:physicality condition} for operations $T \in M_\ell(A_3) = M_\ell(\FF[x]/(x^2))$ becomes that $T^t T = 0,$ in other words, that $T$ is an orthogonal matrix. Therefore, the group of transversal gates in this case is $O(\ell, \FF[x]/(x^2))$, the $\ell \times \ell$ orthogonal matrices over $\FF[x]/(x^2)$.

    We see that $\begin{pmatrix} I & X \\ X & I \end{pmatrix}$ is orthogonal since $X^2 = 0$, so this is a valid non-trivial two qubit gate. Replacing the block matrices $I$ and $X$ with the corresponding $2\times 2$ matrices, the Clifford Tableau of this gate is $\begin{pmatrix} 1 & 0 & 1 & 1 \\ 0 & 1 & 1 & 1 \\ 1 & 1 & 1 & 0 \\ 1 & 1& 0& 1  \end{pmatrix}$. This is the ``$Y$-controlled-$Y$" gate
\begin{center}
\includegraphics{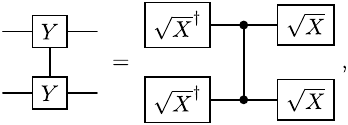}
\end{center}
    which is maximally entangling since it is LDC-equivalent to the $CZ$ gate.
\end{example}

Prior to undertaking this research, we assumed that a code had an entangling transversal two-qubit Clifford gate if and only if it is CSS (i.e. a CNOT). However, the self-dual codes above provide an example of an interesting transversal two qubit gate that exists for codes which are non-CSS. And in fact, our classification theorem tells us that these are the only two cases in which a transversal, entangling two-qubit Clifford exists. We summarize this as a corollary. 

\begin{corollary}
    A code $C$ has a transversal, entangling two qubit Clifford gate if and only if it is LDC-equivalent to a CSS code or a self-dual code (or both).
\end{corollary}

\begin{proof}
    From Theorem \ref{algbera classification}, we know that a code either has the generic endomorphism algebra $A_5$, or is LDC-equivalent to a code which is GF(4)-linear (endomorphism algebra $A_1$), CSS (algebras $A_2$, $A_4,$ or $A_0$), or self-dual (algebras $A_0$ or $A_3$). In the $GF(4)$-linear case, we know there is no entangling gate by Corollary \ref{cor:no tq gate}. In the generic algebra case, $O(2,\FF)$ only has 2 elements, corresponding to the identity and swap. 
    
    This shows that if a code is not LDC-equivalent to CSS or self-dual, there is no entangling gate. However, if a code is CSS or self-dual, then it has a non-trivial two qubit gate: the CNOT or the $Y$-controlled-$Y$ gate respectively.
\end{proof}

We also discuss the generic case, because, even in this case, Theorem \ref{thm:ell block transversal algebra} yields interesting restrictions on what transversal Cliffords can exist in the absence of endomorphisms.

\begin{example}
    Suppose the endomorphism algebra of $C$ is the trivial one, $A_5.$ Then Theorem \ref{thm:ell block transversal algebra} tells us that all transversal gates must be block matrices with blocks either $I$ or $0.$ For instance, the 4 qubit Clifford gate which is transversal for any code given in \cite{Gottesman1997} can be represented as 
    \begin{equation}
        \begin{pmatrix} I & I & I & 0 \\ 0 & I & I & I \\ I & 0 & I & I \\ I &  I & 0& I  \end{pmatrix}.
    \end{equation}
    Since the conjugation operation on $A_5$ fixes the identity,
the symplectic condition \eqref{eq:physicality condition}
simply becomes that if we replace the $I$ and $0$ matrices
by $1$s and $0$s, then the matrix needs to be orthogonal
over $\FF.$ The representation of the above matrix in
terms of the orthogonal matrix group over $\FF$, and
the observation that if the endomorphism
algebra of $C$ is trivial, it at least has this orthogonal
group of gates was first done in \cite{Rains1999}, and
we are just restating it in terms of our notation here. 
\end{example}
However, our theorem yields, moreover, that a code with trivial endomorphism algebra has at most the orthogonal group over $\FF$, it yields the new result excluding the possibility of 1 to 3 qubit transversal gates for a large class of codes.
\begin{corollary}
    Assume $C$ is a stabilizer code which is not LDC-equivalent to a CSS code, $GF(4)$-linear code, or self-dual code. Then it has no non-trivial one qubit, two qubit, or three qubit transversal Clifford gates.
\end{corollary}
\begin{proof}
    In this case, the endomorphism algebra of the code $C$ must be the trivial one. If we look at the size of $O_\ell(\FF)$ in our table on the first page, we see that it is of size $1,2$ and $6$ for $\ell = 1, 2,$ and $3$ respectively, implying that the only valid transversal Cliffords are permuting the $\ell$ qubits for $\ell \leq 3.$
\end{proof}

Finally, we discuss the $U(\ell, R_8)$ case and give an application to magic state distillation related to \cite{dasu2025}.

\begin{example}
    Suppose $C$ has the three-dimensional endomorphism algebra $A_4$, the ring of upper-triangular matrices. This ring is isomorphic to $R_8,$ the smallest non-commutative ring with unity. We can see that the conjugation operator in \eqref{eq:conj2by2} sends a matrix $\begin{pmatrix}x & y \\ 0 & z\end{pmatrix}$ to $\begin{pmatrix}z & y \\ 0 & x\end{pmatrix}$. This operation is linear and reverses the order of multiplication in $R_8$: $\overline{a + b} = \bar{a} + \bar{b}$ and $\overline{ab} = \bar{b}\bar{a}$. These  properties guarantee that if $B$ and $C$ are in $M_\ell(A_4),$ $\overline{B C}^t = \bar{C}^t\bar{B}^t$ meaning that the set of $\ell \times \ell$ matrices in $B \in M_\ell(A_4)$ such that $\bar{B}^t \bar{B} = I$ forms a group. We call this group $U(\ell, R_8)$ since this is a unitarity condition on these matrices. We remark that in this case, not only is the CNOT transversal, but also the controlled-Z gate.
\end{example}

The above example is an elucidation of the ``CSS semi-self dual case," with endomorphism algebra $A_4.$ As shown in Theorem \ref{algbera classification}, this case is LDC-equivalent to the ``semi-CSS self-dual" case with endomorphism algebra $E,$ the ring of $2\times2$ matrices with an even number of $1$ entries. We end this section with an application to magic state preparation and injection for such codes, used in \cite{dasu2025}.

\begin{corollary}\label{cor:magic}
    Suppose $C$ is a code with endomorphism algebra $E$, encoding a single logical qubit. Assuming that the logical action of the transversal gate is the same as the physical gate applied, e.g., that $H_L = H^{\otimes n}$, then the following circuit, which measures the logical Hadamard operator and implements a $R_y(\pi/4)$ rotation, can be implemented using only transversal gates which preserve the code space. The ``$D$" gate represents destabilizers of the code, which may need to be applied after applying an $H^{\otimes n}$ gate, which need only preserve stabilizers up to sign.

\begin{center}
    \includegraphics{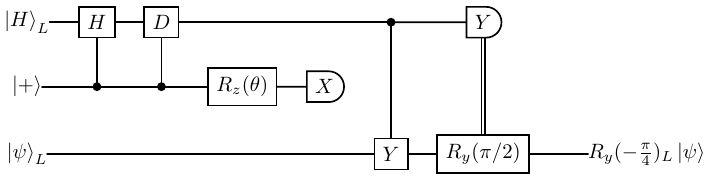}
\end{center}

\end{corollary}

\begin{proof}
    Codes with endomorphism algebra $E$ have a transversal Hadamard, so if the $\ket{+}$ is a physical ancilla, the controlled-Hadamard can be done transversally, in the sense that the same unitary is applied to all qubits of the code, even though the control is a single qubit and not a codeblock. A physical $R_z(\theta)$ may need to be applied to the ancilla before measuring if the $DH$ operator fixes the magic state up to a global phase (see Example \ref{ex:magic}, for instance). These codes also have a controlled-Y gate. To see this, observe that conjugation by $\sqrt{X} H$ will send the transversal CNOT gate enjoyed by codes with endomorphism algebra $A_4$ to a transversal controlled-Y gate. Finally, we observe that $R_y(\frac{\pi}{2}) = HX,$ so this can be done transversally as well.
\end{proof}

\begin{example}\label{ex:magic}
    We now discuss how the above construction yields the magic state preparation and benchmarking protocol that appeared in \cite{dasu2025}, see Fig. 1 in that paper. We describe how the protocol which produces a $\ket{H^{+},Y^{-}}_L$ state encoded in the [[6,2,2]] code is actually using a code with endomorphism algebra $E_1$ on which the circuit in Corollary \ref{cor:magic} can be performed. 

    The [[6,2,2]] code was first introduced as the $C_6$ in \cite{Knill_2005}. However, we will use a later presentation of the stabilizers and logicals, as in \cite{Goto_magic} and \cite{Jones2013}, which was introduced for the purpose of magic state distillation. In this presentation, the stabilizers of the [[6,2,2]] code are generated by $XXXXII, IIXXXX, ZZZZII, IIZZZZ$ and the logical operators are given by $\bar{X}_1 = XIXIXI, \bar{Z}_1 = ZIZIZI, \bar{X}_2 = IXIXIX, \bar{Z_2} = IZIZIZ.$ This is a self-dual CSS code, and, moreover, this representation of the logical operators has the nice property that the transversal Hadamard performs a logical Hadamard on both qubits. 

    In the protocol above, if we view the second logical qubit in the [[6,2,2]] as being gauge-fixed to $\ket{Y^-}_L$, then we can regard the resulting code as a [[6,1,2]] code with stabilizers $XXXXII, IIXXXX, ZZZZII, IIZZZZ, IYIYIY$. This is no longer a CSS code, because of the $Y$ stabilizer but it remains self-dual with an additional ``semi-CSS" endomorphism that sends $X$ to $Y$ and $Z$ to $I.$ Therefore, the endomorphism algebra of this code is $E$. When performing the circuit in Corollary \ref{cor:magic}, the transversal $H$ gate will send the $IYIYIY$ stabilizer to $-IYIYIY,$ so one must apply a controlled-destabilizer of $IYIYIY$ as well, which can be taken to be $IZIZIZ.$ Actually, the application of the transversal Hadamard followed by the $IZIZIZ$ destabilizer, leads to an undesirable phase on the ancilla qubit. More precisely, $\ket{+}\ket{H^+}_L$ is sent to $\ket{0}\ket{H^+}_L + e^{-i\pi/4}\ket{1}\ket{H^+}_L$, so a $T$ gate must be applied to the ancilla before measuring it. If we add a flag to the ancilla, then the circuit is fault-tolerant to a single failure and we obtain the circuit in \cite{dasu2025}. 
\end{example}

\section{Conclusion}

In this work we have completely solved the characterization
of quantum stabilizer code families via their diagonal Clifford symmetries,
or equivalently, their $\Field_2$-linear endomorphism algebras.
This operational approach to characterizing codes
yields a Galois correspondence, 
or a Kleinian geometry~\cite{Klein1893}.

This correspondence yields some surprising consequences. For instance, it seems that the existence of non-CSS codes with a maximally entangling two-qubit gate is not widely known and there may be non-trivial applications for such codes, such as the magic state protocol in Example \ref{ex:magic}.

A question we leave open is the physical interpretation
of the non-invertible $\Field_2$-linear
endomorphisms supported by $C^{(\ell)}$.



{\bf Acknowledgements}
We thank Selwyn Simsek, Elijah Durso-Sabina 
and Natalie Brown for helpful discussion and feedback.

\bibliography{refs}{}
\bibliographystyle{abbrv}

\end{document}